\documentclass[onecolumn]{IEEEtran}
\usepackage{algorithm}
\usepackage{algpseudocode}
\usepackage{times,amssymb,amsmath,amsfonts,nicefrac,float,accents,color,graphicx,bbm,
caption,subcaption,mathrsfs,amsthm,stmaryrd}
\usepackage{relsize}
\usepackage[mathscr]{eucal}
\newcommand\remove[1]{}

\allowdisplaybreaks

\newtheorem{theorem}{Theorem}

\newtheorem{lemma}[theorem]{Lemma}
\newtheorem{proposition}[theorem]{Proposition}

\theoremstyle{remark}
\newtheorem{definition}{Definition}[section]
\theoremstyle{remark}
\newtheorem{remark}{Remark}[section]

\newtheorem{example}{Example}

\newcommand\ff{{\mathbb F}}
\newcommand{\ba}{\begin{array}}
\newcommand{\ea}{\end{array}}
\newcommand{\be}{\begin{equation}}
\newcommand{\ee}{\end{equation}}
\newcommand{\bea}{\begin{eqnarray}}
\newcommand{\eea}{\end{eqnarray}}

\def\argmax{\text{argmax}}
\newcommand\nc\newcommand
\nc\bfa{{\boldsymbol a}}\nc\bfA{{\boldsymbol A}}\nc\cA{{\mathscr A}}
\nc\bfb{{\boldsymbol b}}\nc\bfB{{\boldsymbol B}}\nc\cB{{\mathscr B}}
\nc\bfc{{\boldsymbol c}}\nc\bfC{{\boldsymbol C}}\nc\cC{{\mathscr C}}
\nc\sC{{\mathscr C}}
\nc\bfd{{\boldsymbol d}}\nc\bfD{{\bfD}}
\nc\cD{{\mathscr D}}
\nc\bfe{{\boldsymbol e}}\nc\bfE{{\boldsymbol E}}\nc\cE{{\mathscr E}}
\nc\bff{{\boldsymbol f}}\nc\bfF{{\boldsymbol F}}\nc\cF{{\mathscr F}}
\nc\bfg{{\boldsymbol g}}\nc\bfG{{\boldsymbol G}}\nc\cG{{\mathscr G}}
\nc\bfh{{\boldsymbol h}}\nc\bfH{{\boldsymbol H}}\nc\cH{{\mathscr H}}
\nc\bfi{{\boldsymbol i}}\nc\bfI{{\boldsymbol I}}\nc\cI{{\mathscr I}}\nc\sI{{\mathscr I}}
\nc\bfj{{\boldsymbolj}}\nc\bfJ{{\boldsymbol J}}\nc\cJ{{\mathscr J}}
\nc\bfk{{\boldsymbolk}}\nc\bfK{{\boldsymbol K}}\nc\cK{{\mathscr K}}
\nc\bfl{{\boldsymboll}}\nc\bfL{{\boldsymbol L}}\nc\cL{{\mathscr L}}
\nc\bfm{{\boldsymbolm}}\nc\bfM{{\boldsymbol M}}\nc\cM{{\mathscr M}}
\nc\bfn{{\boldsymboln}}\nc\bfN{{\boldsymbol N}}\nc\cN{{\mathscr N}}
\nc\bfo{{\boldsymbolo}}\nc\bfO{{\boldsymbol O}}\nc\cO{{\mathscr O}}
\nc\bfp{{\boldsymbolp}}\nc\bfP{{\boldsymbol P}}\nc\cP{{\mathscr P}}
\nc\eP{{\EuScriptP}}\nc\fP{{\mathfrak P}}
\nc\bfq{{\boldsymbol q}}\nc\bfQ{{\boldsymbol Q}}\nc\cQ{{\mathscr Q}}
\nc\bfr{{\boldsymbol r}}\nc\bfR{{\boldsymbol R}}\nc\cR{{\mathscr R}}
\nc\bfs{{\boldsymbol s}}\nc\bfS{{\boldsymbol S}}\nc\cS{{\mathscr S}}
\nc\bft{{\boldsymbol t}}\nc\bfT{{\boldsymbol T}}\nc\cT{{\mathscr T}}
\nc\bfu{{\boldsymbol u}}\nc\bfU{{\boldsymbol U}}\nc\cU{{\mathscr U}}
\nc\bfv{{\boldsymbol v}}\nc\bfV{{\boldsymbol V}}\nc\cV{{\mathscr V}}

\nc\bfw{{\boldsymbol w}}\nc\bfW{{\bf W}}\nc\cW{{\mathcal W}}\nc\sW{{\mathscr W}}
\nc\bfx{{\boldsymbol x}}\nc\bfX{{\boldsymbol X}}\nc\cX{{\mathcal X}}\nc\sX{{\mathscr X}}
\nc\bfy{{\boldsymbol y}}\nc\bfY{{\boldsymbol Y}}\nc\cY{{\mathcal Y}}\nc\sY{{\mathscr Y}}
\nc\bfz{{\boldsymbol z}}\nc\bfZ{{\boldsymbol Z}}\nc\cZ{{\mathcal Z}}\nc\sZ{{\mathscr Z}}

\newcommand{\bfit}{\bfseries\itshape}

\renewcommand{\oplus}{+}

\begin{document}

\title{Construction of polar codes for arbitrary discrete memoryless channels\thanks{The authors are with Dept. of ECE and ISR, University of Maryland, College Park, MD 20742, USA. Emails: \{tcgulcu,yeemmi\}@gmail.com, abarg@umd.edu. A. Barg is also
with Inst. Probl. Inform. Trans. (IITP), Moscow, Russia. Research supported 
in part by
NSF grants CCF1217245 and CCF1422955.}}

\author{
Talha Cihad Gulcu
\hspace{0.5cm} 
\and 
Min Ye
\hspace{0.5cm}
\and 
Alexander Barg
}

\maketitle

\begin{abstract} 
It is known that polar codes can be efficiently constructed for binary-input channels. At the same time, existing
algorithms for general input alphabets are less practical because of high complexity. We address the construction 
problem for the general case, and analyze an algorithm that is based on successive reduction
of the output alphabet size of the subchannels in each recursion step. For this procedure we estimate the approximation error
as $O(\mu^{-1/(q-1)}),$ where $\mu$ is the ``quantization parameter,'' i.e., the maximum size of the subchannel output alphabet
allowed by the algorithm. The complexity of the code construction scales as $O(N\mu^2 \log \mu),$ where $N$ is the length of the code.

We also show that if the polarizing operation relies on modulo-$q$ addition, it is possible to merge subsets of output
symbols without any loss in subchannel capacity. Performing this procedure before each approximation step results in 
a further speed-up of the code construction, and the resulting codes have smaller gap to capacity. We also show that a similar 
acceleration can be attained for polar codes over finite field alphabets. 

Experimentation shows that the suggested construction algorithms can be used to construct long polar codes 
for alphabets of size $q=16$ and more with acceptable loss of the code rate for a variety of polarizing transforms.
\end{abstract}
\maketitle

\noindent{\bfit Index terms:} Channel degrading, Greedy symbol merging, Polarizing transforms.

\section{Introduction}

Ar{\i}kan's polar codes \cite{arikan2009} form the first explicit family of binary codes that achieve
the capacity of binary-input channels. 
Polar codes rely on a remarkable phenomenon called channel polarization. 
After their introduction, both polar codes and the channel polarization concept have been used in a vast range of problems in information theory \cite{abbe}--\cite{hassani}.
While \cite{arikan2009} described efficient encoding and decoding procedures of polar codes, it also noted that their construction 
presents a difficult algorithmic challenge because the alphabet of the bit subchannels grows exponentially as a function of the number of iterations of the polarization procedure. Approached straightforwardly, this results in an exponential complexity of the code construction.

The difficulty of selecting subchannels for information transmission with polar codes was recognized early on in a number of papers. 
According to an observation made in \cite{mori_tanaka2}, 
the construction procedure of polar codes for binary-input channels relies on essentially the same density evolution
procedure that plays a key role in the analysis of low-density parity-check codes.
It was soon realized that the proposal of \cite{mori_tanaka2} requires increasing precision of the computations, but this paper paved way for later research on the construction problem.

An important step was taken in \cite{tal_vardy} which suggested 
to approximate each bit-channel after each evolution step by its degraded or upgraded version whose output alphabet size is constrained by a specified threshold $\mu$ that serves as a parameter of the procedure. As a 
result, \cite{tal_vardy} put forward an approximation procedure that results in a code not too far removed from
the ideal choice of the bit-channels of \cite{arikan2009}. This code construction scheme has a complexity of $O(N \mu^2 \log \mu)$, where $N=2^n$ is the code length. 
For the channel degradation method described in \cite{tal_vardy}, an error analysis and approximation guarantees are provided in \cite{pedarsani}. 

Another approximation scheme for the construction of binary codes was considered in \cite{sasoglu5}.
It is based on degrading each bit-channel after each evolution step, performed by
merging several output symbols into one symbol based on quantizing the curve 
$p_{X|Y} (0|y)$ vs $h(p_{X|Y} (0|y)),$ where $p_{X|Y}$ is the conditional distribution of the ``reverse channel'' 
that corresponds to the bit-channel in question.  
Symbols of the output alphabet that share the same range of quantization are merged into a single symbol of the approximating channel.
Yet another algorithm based on bit-channel upgrading was described in \cite{ghayoori_gulliver}, 
in which the authors argue that it is possible to obtain a channel which is arbitrarily close to the bit-channel of interest in terms of the capacity. However, no error or complexity analysis is provided in this work.

Moving to general input alphabets, let us mention a code construction algorithm 
based on degrading the subchannels in each evolution step designed in \cite{tal_sharov_vardy}.
This algorithm involves a merging procedure of output symbols similarly to \cite{sasoglu5},
having a complexity $O(N \mu^2 \log \mu)$.
However, as noted by the authors, the construction scheme of \cite{tal_sharov_vardy} is practical only for small values of input alphabet size $q$ because the size of the output alphabet $\mu$ has to be of the form $\mu=\lambda^q$, where $\lambda$ is the number of quantiazation levels used in the merging process. Paper \cite{pereg_tal} proposed to 
perform the upgrading instead of degrading of the subchannels, but did not manage to overcome the
limitation imposed by the constraint $\mu=\lambda^q$.
In \cite{ghayoori_gulliver2}, the authors consider another channel 
upgrading method for nonbinary-input channels, but stop short of providing an explicit construction scheme or error analysis.

Papers \cite{trifonov,li_yuan,wu_li_sun} addressed the 
construction problem of polar codes for AWGN channels.  These works
are based on Gaussian approximation of the intermediate likelihood ratios and do not analyze
the error guarantees or rate loss of the obtained codes. A comparative study of various polar code 
constructions for AWGN channel is presented
in \cite{vangala}.
Some other heuristic constructions for binary-input channels similar to the cited
results for the Gaussian channel appear in \cite{kern,bonik,zhao}.
Note also constructions of
polar codes for some particular channels \cite{santos,chen},
for various transformation kernels \cite{zhang,trifonov2,serbetci}, and concatenated codes \cite{trifonov3,mahdavifar2}. 

In this paper we present a construction method of polar codes for input alphabets of arbitrary size, together 
with explicit analysis of approximation error and construction complexity.
\textcolor{black}{In particular, the complexity estimate of our procedure grows as $O(N\mu^2\log\mu)$.} 
Our algorithm can be viewed as a generalization of the channel degradation
method in \cite{tal_vardy} to nonbinary input channels. Although the approach and the proof methods here are rather 
different from earlier works, the estimate of the approximation error that we derive generalizes the error bound given
by \cite{pedarsani} for the binary case. Another interesting connection with the literature concerns
a very recent result of \cite{tal} which derives a lower bound on the alphabet size $\mu$ that is necessary
to restrict the capacity loss by at most a given value $\epsilon.$ This bound is valid for
any approximation procedure that is based only on the degrading of the subchannels in each evolution step.
The construction scheme presented here relies on the value $\mu$ that
is not too far from this theoretical limit (see Proposition \ref{third_lemma} for more details).
We stress that we aim at approximating symmetric capacity of the channels, and do not attempt to construct or implement polar codes
that attain Shannon capacity, which is greater than the symmetric one for non-symmetric channels.
 
 Our paper is organized as follows. In Section  \ref{prelim} we give a brief overview of polar codes 
including various polarizing transformations for nonbinary alphabets. The 
rate loss estimate in the code construction based on 
merging pairs of output symbols in a greedy way is derived in Section \ref{main}. 
In Section \ref{cyclic} we argue that
output symbols whose posterior probability vectors are cyclic shifts of each other can be merged with
no rate loss. This observation enables us to formulate an improved version of the construction algorithm that
further reduces the construction complexity. We have also implemented our algorithms and constructed
polar codes for various nonbinary alphabets. These results are presented in Section \ref{sim}. 
For relatively small $q$ we can construct rather long polar codes
(for instance, going to length $10^6$ for $q=5$ takes several hours). 
For larger $q$ such as 16 we can reach lengths of tens of thousands within reasonable time and with low rate loss.
Even in this case, by increasing the gap to capacity of the resulting codes, we can reach lengths in the range of hundreds of thousands to a million without putting an effort in optimizing our software.

\section{Preliminaries on Polar Coding}
\label{prelim}

We begin with a brief overview of binary polar codes.
Let $W$ be a channel with the output alphabet ${\mathscr Y},$ input alphabet
${\mathscr X}=\{0,1\},$ and the conditional probability distribution $W(\cdot|\cdot)=W_{Y|X}(\cdot|\cdot).$ Throughout the paper we denote the capacity and  the symmetric capacity of $W$ by $C(W)$ and $I(W)$, respectively. We say
$W$ is symmetric if $W (y|1), y\in {\mathscr Y}$ can be obtained from $W (y|0), y\in {\mathscr Y}$  
through a permutation $\pi:\sY\to\sY$ such that $\pi^2$ \textcolor{black}{is the identity mapping.} 
Note that if $W$ is symmetric then $I(W)=C(W).$

For $N=2^n$ and $n\in{\mathbb N}$, define the polarizing matrix (or the Ar{\i}kan transform matrix)
as $G_N=B_N F^{\otimes n}$, where $F=\text{\small{$\Big(\hspace*{-.05in}\begin{array}{c@{\hspace*{0.05in}}c}
    1&0\\[-.05in]1&1\end{array}\hspace*{-.05in}\Big)$}}$,
$\otimes$ is the Kronecker product of matrices, and $B_N$ is a 	``bit reversal'' permutation 
matrix \cite{arikan2009}. In \cite{arikan2009}, Ar{\i}kan showed that given a 
symmetric and binary input channel $W$, an appropriate subset of the rows of $G_N$ can be used as a generator matrix of a linear code that achieves the capacity of $W$ as $N\to\infty$.

Given a binary-input channel $W$, define the channel $W^N$ with input alphabet
$\{0,1\}^N$ and output alphabet ${\mathscr Y}^N$ by the conditional distribution
\begin{equation*}
W^N(y^N|x^N)= \prod_{i=1}^N W(y_i|x_i)
\end{equation*}
where \textcolor{black}{$W(\cdot|\cdot)$} is the conditional distribution that defines $W$.
Define a combined channel $\widetilde{W}$ by the conditional distribution
\begin{equation*}
\widetilde{W}(y^N|u^N)= W^N(y^N|u^N G_N).
\end{equation*}
In terms of $\widetilde{W}$, the channel seen by the $i$-th bit $U_i, i=1,\dots,N$
(also known as the bit-channel of the $i$-th bit) can be written as
\begin{equation}\label{eq:W}
\textcolor{black}{W_N^{(i)}} (y^N, u_1^{i-1}|u_i)= \frac{1}{2^{n-1}} \sum_{\widetilde{u}\in \{0,1\}^{n-i}} \widetilde{W} (y^N| (u_1^{i-1},u_i, \widetilde{u})).
\end{equation}
We see that $W_i$ is the conditional distribution of $(Y^N, U_1^{i-1})$ given $U_i$
provided that the channel inputs $X_i$ are uniformly distributed for all $i=1,\dots,N$.
Moreover, it is the case that \cite{arikan2009} the bit-channels $W_i$ can be
constructed recursively using the channel transformations $W^{-}$ and $W^{+}$,
which are defined by the equations
\begin{align}
W^{-}(y_1,y_2|u_1)&\triangleq\frac{1}{2}\sum_{u_2\in\{0,1\}} W(y_1|u_1\oplus u_2) W(y_2|u_2) \label{eq:+}\\
W^{+}(y_1,y_2,u_1|u_2) &\triangleq \frac{1}{2} W(y_1|u_1\oplus u_2) W(y_2|u_2).\label{eq:-}
\end{align}

\remove{Given a binary RV $X$ and a discrete RV $Y$ supported on ${\mathscr Y}$,
define the Bhattacharyya parameter $Z(X|Y)$ as 
$
Z(X|Y)=2\sum_{y\in{\mathscr Y}} P_Y(y) \sqrt{{P_{X|Y}(0|y) P_{X|Y}(1|y)}}.
$
The value $Z(X|Y), 0\le Z(X|Y)\le 1$ measures the amount of randomness in $X$ given $Y$ in the sense
that if it is close to zero, then $X$ is almost constant given $Y$, while if it is close to one, then $X$ is almost uniform on $\{0,1\}$ given $Y$.}
The Bhattacharyya parameter $Z(W)$ of a binary-input channel $W$ is defined as
$
Z(W)=\sum_{y\in {\mathscr Y}} \sqrt{W_{Y|X}(y|0) W_{Y|X}(y|1)}.
$
The bit-channels defined in \textcolor{black}{\eqref{eq:+}--\eqref{eq:-}} are partitioned into good channels ${\mathscr G}_N (W,\beta)$
and bad channels ${\mathscr B}_N (W,\beta)$ based on the values of $Z(\textcolor{black}{W_N^{(i)}})$.
More precisely, we have
\begin{equation}\label{good-bad}
\begin{aligned}
{\mathscr G}_N (W,\beta)&= \{ i\in [N]: Z(W_N^{(i)}) \leq 2^{-N^{\beta}} \}\\
{\mathscr B}_N (W,\beta)&= \{ i\in [N]: Z(W_N^{(i)})  > 1-2^{-N^{\beta}} \} ,
\end{aligned}
\end{equation}
where $[N]=\{1,2,\dots,N\}.$ 
As shown in \cite{ari09a}, for any binary-input channel $W$ and
any constant $\beta<1/2,$
      \begin{equation}\label{telatar}
   \begin{aligned}
     \lim_{N\to\infty} \frac{|{\mathscr G}_N (W,\beta)|}{N}&= I(W)\\
     \lim_{N\to\infty} \frac{|{\mathscr B}_N (W,\beta)|}{N}&= 1-I(W).
     \end{aligned}
   \end{equation}
Based on this equality, information can be transmitted over the \textcolor{black}{good bit-channels} while the remaining
bits are fixed to some values known in advance to the receiver (in polar coding literature they are called \emph{frozen bits}).  
The transmission scheme can be described as follows: 
A message of $k=|{\mathscr G}_N (W,\beta)|$ bits is written in the bits $u_i, i\in {\mathscr G}_N (W,\beta).$ 
The remaining $N-k$ bits are set to 0. This determines the sequence $u^N$ which
is transformed into $x^N=u^N G_N,$ and the vector $x^N$ is sent over the channel. Denote by $y^N$
the sequence received on the output.
The decoder finds an estimate of $u^N$ by computing the values $\hat u_i, i=1,\dots,N$ as follows:
\begin{align}\label{eq:u}
\hat{u}_i=
\begin{cases}
\argmax_{u\in \{0,1\}}W_i (y^N,\hat{u}_1^{i-1}|u), &\text{if} \, i\in {\mathscr G}_N (W,\beta),  \\
0, &\text{if} \, i\in {\mathscr B}_N (W,\beta).
\end{cases}
\end{align}
The results of \cite{arikan2009,ari09a} imply the following upper bound on the error probability $P_e=\Pr(\hat u^N\ne u^N):$
\begin{equation}\label{eq:ep}
P_e \leq \sum_{i\in {\mathscr G}_N (W,\beta)} Z(W_N^{(i)})  \leq N 2^{-N^{\beta}}
\end{equation}
where $\beta=\frac12-\epsilon,$ and $\epsilon>0$ is arbitrarily small. 
This describes the basic construction of polar codes \cite{arikan2009} which  
attains symmetric capacity $I(W)$ of the channel $W$ with a low error rate.
At the same time, \eqref{eq:W}, \eqref{eq:u} highlight the main obstacle in the way of efficiently constructing polar codes: the size of the output alphabet of the channels $W_i$ is of the order \textcolor{black}{$|\sY|^{N},$} so it scales exponentially with the code length.
For this reason, finding a practical code construction scheme of polar codes represents a nontrivial problem.

Concluding the introduction, let us mention that the code construction technique presented below can be applied to any polarizing transform
based on combining pairs of subchannels. There has been a great deal of 
research on properties of polarizing operations in general. 
In particular, it was shown in \cite{sasoglu3} that \eqref{eq:ep} holds true whenever the 
input alphabet size $q$ of the channel $W$ is a prime number, and $W^{-}$ and $W^{+}$
are defined as
\begin{align}
W^{-}(y_1,y_2|u_1)&\triangleq\frac{1}{q}\sum_{u_2\in\{0,1,\dots,q-1\}} \!\!\!\!W(y_1|u_1\oplus u_2) W(y_2|u_2) \label{eq:+q}\\
W^{+}(y_1,y_2,u_1|u_2) &\triangleq \frac{1}{q} W(y_1|u_1\oplus u_2) W(y_2|u_2),\label{eq:-q}
\end{align}
meaning that
Ar{\i}kan's transform 
$F=\text{\small{$\Big(\hspace*{-.05in}\begin{array}{c@{\hspace*{0.05in}}c}
1&0\\[-.05in]1&1\end{array}\hspace*{-.05in}\Big)$}}$ 
 is polarizing for prime alphabets. For the case when $q$ is a power of a prime, it
was proved in \cite{mori_tanaka} that there exist binary linear transforms different from 
$F$ that support the estimate in \eqref{eq:ep} for some exponent $\beta$ that depends on $F$.
For example, \cite{mori_tanaka} shows that the transform  
\remove{$G_{\gamma}=\text{\small{$\Big(\hspace*{-.05in}\begin{array}{c@{\hspace*{0.05in}}c}
   1&0\\-.05in]\gamma&1\end{array}\hspace*{-.05in}\Big)$}}$}
   \begin{equation}\label{eq:mt}
     G_\gamma=\begin{pmatrix} 1&0\\ \gamma&1
     \end{pmatrix}
   \end{equation}
is polarizing whenever $\gamma$ is a primitive element of the field ${\mathbb F}_q$.
Paper \cite{park_barg} considered the use of Ar{\i}kan's transform for the
channels with input alphabet of size $q=2^r$, showing that
the symmetric capacities of the subchannels converge to one of $r+1$ integer values in the set
$\{0,1,\dots,r\}.$ 

Even more generally, necessary and sufficient conditions for a binary operation $f:{\mathcal X}^2\to{\mathcal X}^2$ given by
\begin{align}\label{eq:ff}
u_1&=f(x_1,x_2), \\
u_2&=x_2.\nonumber
\end{align}
to be a polarizing mapping were identified in \cite{nasser1}. A simple set of 
{\em sufficient} conditions for the same was given in \cite{sasoglu4}, which also gave a
concrete example of a polarizing mapping for an alphabet of arbitrary size $q.$ According to
\cite{sasoglu4}, in \eqref{eq:ff} one can take $f$ in the form
$f(x_1,x_2)= x_1+\pi (x_2)$, where $\pi: {\mathcal X} \to {\mathcal X}$
is the following permutation:
\begin{align}
\pi(x)= \begin{cases}
\lfloor q/2  \rfloor,   &\text{if} \quad  x=0, \\
x-1, & \text{if} \quad 1\leq x \leq \lfloor q/2  \rfloor, \\
x, & \text{otherwise}.
\end{cases}\label{eq:pi}
\end{align}
We include experimental results for code construction using the transforms \eqref{eq:mt} and \eqref{eq:pi} 
in Sect.~\ref{sim}.

Finally recall that it is possible to attain polarization based on transforms that combine $l> 2$
subchannels. In particular, polarization results for transformation kernels of 
size $l\times l$ with $l>2$ for binary-input channels were studied in \cite{korada2}. 
Apart from that,  \cite{mori_tanaka} derived estimates of the error probability of polar codes 
for nonbinary channels based on transforms defined by generator matrices of Reed-Solomon codes. 
However, below we will restrict our attention to binary combining operations of the form discussed above.

\section{Channel Degradation and the Code Construction Scheme}
\label{main}

In the algorithm that we define, the subchannels 
are constructed recursively, and after each evolution step
the resultant channel is replaced by its degraded version which has an output alphabet
size less than a given threshold $\mu$. In general terms, this procedure is described in more detail as follows.

\begin{algorithm}
\caption{Degrading of subchannels}\label{euclid}
\hspace{-0.0in}\textbf{input:} DMC $W$, bound on the output size $\mu$, code length $N=2^n$, 
channel index $i$ with binary representation 

\hspace{-0.0in}$i=\langle b_1,b_2,\dots b_n\rangle_2.$

\hspace{-0.0in}\textbf{output:} A DMC obtained from the subchannel $W_N^{(i)}$.
\begin{algorithmic}
\State $T_N^{(i)}$ $\gets \texttt{degrade} (W,\mu)$
\For{$j=1,2,\dots,n$} 
\If{ $b_j=0$ }
\State $T_N^{(i)}$ $\gets T^{-}$
\State \textbf{else}
\State $T_N^{(i)}$ $\gets T^{+}$
\EndIf
\State $T_N^{(i)}$ $\gets \texttt{degrade}(T,\mu)$
\EndFor
\State \textbf{return} $T_N^{(i)}$
\end{algorithmic}
\end{algorithm}

Before proceeding further we note that  $T^{-}$ and $T^{+}$ appearing in Algorithm \ref{euclid} can be any transformations
that produce combined channels for the polarization procedure. The possibilities range from Ar{\i}kan's transform
to the schemes discussed in the end of Section \ref{prelim}.

The next step is to define the function $\texttt{degrade}$ in such a way that it 
can be applied to general discrete channels. 
Ideally, the degrading-merge operation 
should optimize the 
degraded channel by attaining the smallest rate loss over all $T':$ 
\begin{equation}
\inf_{ \substack{ T': \,T' \prec W \\  |\text{out}(T')| \leq \mu  } } I(W)-I(T') \label{con_max}
\end{equation}
Equation \eqref{con_max} defines a convex maximization problem, which is difficult to solve
with reasonable complexity. To reduce the computational load, \cite{tal_vardy} proposed the following approximation
to \eqref{con_max}: replace $y,y'\in \sY$ by a single symbol
if the pair $y,y'$ gives the minimum loss of capacity among all pairs of output symbols, and repeat this
as many times as needed until the number of the remaining output symbols is equal to or 
less than $\mu$ (see Algorithm C in \cite{tal_vardy}). In \cite{pedarsani,sasoglu5} this procedure was called {\em greedy mass merging}. 
In the binary case this procedure can be implemented with complexity $O(N \mu^2 \log \mu)$ because one
can check only those pairs of symbols $(y_1,y_2)$ which  are closest to each other in terms
of the likelihood ratios (see Theorem 8 in \cite{tal_vardy}), \textcolor{black}{and there are
$$N+\frac{N}{2}+\dots+2=2N-2=O(N)$$ virtual channels in total for which this procedure needs to be carried out.} This simplification does not
generalize to the channels with nonbinary inputs, meaning that we need to inspect
all pairs of symbols. Since the total number of pairs is $O(\mu^4)$ after each evolution step, 
the overall complexity of the greedy mass merging algorithm for nonbinary input alphabets is at most $O(N \mu^4 \log \mu)$.

\textcolor{black}{There is a faster way to perform the search for closest pairs in a metric space \cite{BentleyShamos}, relying on which
the complexity of each evolution step can be estimated as $O(\mu^2 \log \mu),$ \textcolor{black}{although the implicit constant
grows rapidly with the size of the output alphabet.} Thus, the overall complexity of our algorithm is 
$O(N \mu^2 \log \mu)$, as claimed earlier.
} 

For a channel $W:\sX \to \sY$ define 
\begin{align*}
P_W(x|y)&=\frac{W(y|x)}{\sum_{x_0\in\sX}W(y|x_0)},\\
P_Y(y)&=\frac{1}{q}\sum_{x_0\in\sX}W(y|x_0)
\end{align*}
for all $x\in\sX$ and $y\in\sY.$ For a subset $A\subseteq\sY,$ define
$$
P_Y(A)=\sum_{y\in A}P_Y(y).
$$

In the following lemma we establish
an upper bound on the rate loss of the greedy mass merging algorithm for nonbinary input alphabets.

\begin{lemma}\label{lemma:defmerge}
 Let $W:{\mathscr X}\to{\mathscr Y}$ be a discrete memoryless channel and 
let $y_1, y_2\in \sY$ be two output symbols. Let 
$\tilde{W}: {\mathscr X}\to {\mathscr Y}\backslash  \{y_1,y_2\} \cup\{y_\text{merge}\}$ be the
channel that is obtained from by $W$ \textcolor{black}{by replacing $y_1$ and $y_2$ with a new symbol $y_\text{merge}$} and that has the transition probabilities
\begin{align*}
\tilde{W}(y|x)=\begin{cases}
W(y|x), &\text{if} \quad y\in {\mathscr Y}\backslash \{y_1,y_2\} \\
W(y_1|x)+W(y_2|x), &\text{if} \quad y=y_\text{merge}
\end{cases}.
\end{align*}
Then
\begin{equation}
0\leq I(W)-I(\tilde{W}) \leq \frac{P_Y(y_1)+ P_Y(y_2)}{\ln 2}\, \sum_{x\in\sX} |P_W(x|y_1)-P_W(x|y_2)|  .\label{merge_ineq}
\end{equation}

\label{first_lemma}
\end{lemma}
\begin{proof}
Since $\tilde{W}$ is degraded with respect to $W$, we clearly have that $I(W)\ge I(\tilde{W}),$ where $I(\cdot)$ is the 
symmetric capacity.
To prove the upper bound for $I(W)-I(\tilde{W})$ in \eqref{merge_ineq} let $X$ be the random variable uniformly distributed
on ${\mathscr X}$, and let $Y$ be the random output of $W$. Then we have
   \begin{align}
I(W)-I(\tilde{W})& = \Big(H(X)- \sum_{y\in{\mathscr Y}}H(X|Y=y) P_Y(y) \Big)\nonumber\\[.1in]
- \Big( H(X)
-&H(X|Y\in\{y_1,y_2\}) (P_Y(y_1)+P_Y(y_2))-  \!\!\!\!\!\!\!\! \sum_{y\in{\mathscr Y}\backslash \{y_1,y_2\} } \!\!\!\!\!\!\!\! 
H(X|Y=y) P_Y(y) \Big)
    \nonumber\\[.1in]
&= H(X|Y\in\{y_1,y_2\}) (P_Y(y_1)+P_Y(y_2))\nonumber\\&\hspace*{1in}- H(X| Y=y_1) P_Y(y_1)- H(X|Y=y_2) P_Y(y_2). \label{ineq1}
    \end{align}
Next we have
    \begin{align*}
\Pr(X=x|Y\in\{y_1,y_2\})&= \frac{ \frac{1}{|{\mathscr X}|} (W(y_1|x)+W(y_2|x)) }{P_Y(y_1)+P_Y(y_2) } \\
&= \frac{ \frac{1}{|{\mathscr X}|} W(y_1|x) }{P_Y(y_1)+P_Y(y_2) } + \frac{ \frac{1}{|{\mathscr X}|} W(y_2|x) }{P_Y(y_1)+P_Y(y_2) } \\
&=\frac{P_Y(y_1)}{P_Y(y_1)+P_Y(y_2)} P_W (x|y_1) + \frac{P_Y(y_2)}{P_Y(y_1)+P_Y(y_2)} P_W (x|y_2) \\
&=\alpha_{12} P_W (x|y_1)+ (1-\alpha_{12}) P_W (x|y_2)
    \end{align*}
where $\alpha_{12}\triangleq \frac{P_Y (y_1)}{P_Y(y_1)+P_Y(y_2)}$. Hence, it follows from \eqref{ineq1} that
   \begin{align*}
I(W)-I(\tilde{W})= &(P_Y(y_1)+P_Y(y_2)) \sum_{x\in{\mathscr X}} \left[\alpha_{12} P_W (x|y_1)+ (1-\alpha_{12}) P_W (x|y_2) \right] \\
    &\times\log_2 \frac{1}{\alpha_{12} P_W (x|y_1)+ (1-\alpha_{12}) P_W (x|y_2)} \\
    &
- P_Y(y_1) \sum_{x\in{\mathscr X}} P_W (x|y_1) \log_2 \frac{1}{P_W (x|y_1)}
- P_Y(y_2) \sum_{x\in{\mathscr X}} P_W (x|y_2) \log_2 \frac{1}{P_W (x|y_2)}.
     \end{align*}
Rearranging the terms, we obtain
\begin{align*}
    I(W)-I(\tilde{W})= &P_Y(y_1) \sum_{x\in{\mathscr X}} P_{W}(x|y_1) \log_2 
\frac{P_{W}(x|y_1)}{\alpha_{12} P_{W}(x|y_1)+ (1-\alpha_{12}) P_{W}(x|y_2)} \\
&\quad+ P_Y(y_2)  \sum_{x\in {\mathscr X}} P_{W}(x|y_2)  \log_2 
\frac{P_{W}(x|y_2)}{\alpha_{12} P_{W}(x|y_1)+ (1-\alpha_{12}) P_{W}(x|y_2)}.
\end{align*}
Next use the inequality $\ln x\leq x-1$ to write
    \begin{align*}
  I(W)-I(\tilde{W})& \leq P_Y(y_1) \sum_{x\in{\mathscr X}}  \frac{P_{W}(x|y_1)}{\ln 2} 
\left(\frac{P_{W}(x|y_1)}{\alpha_{12} P_{W}(x|y_1)+ (1-\alpha_{12}) P_{W}(x|y_2)}-1\right) \\
&\quad+ P_Y(y_2) \sum_{x\in {\mathscr X}}  \frac{P_{W}(x|y_2)}{\ln 2}  
\left(\frac{P_{W}(x|y_2)}{\alpha_{12} P_{W}(x|y_1)+ (1-\alpha_{12}) P_{W}(x|y_2)}-1\right)
    \end{align*}
which simplifies to
\begin{align}
I(W)-I(\tilde{W}) &\leq \frac{P_Y(y_1)}{\ln 2} \sum_{x\in{\mathscr X}} P_{W}(x|y_1) 
\frac{(1-\alpha_{12}) (P_{W}(x|y_1) -P_{W}(x|y_2))}{ \alpha_{12} P_{W}(x|y_1)+ (1-\alpha_{12}) P_{W}(x|y_2) } \nonumber\\
&\quad+  \frac{P_Y(y_2)}{\ln 2} \sum_{x\in{\mathscr X}} P_{W}(x|y_2)
\frac{ \alpha_{12} (P_{W}(x|y_2)- P_{W}(x|y_1))}{ \alpha_{12} P_{W}(x|y_1)+ (1-\alpha_{12}) P_{W}(x|y_2)}. \label{ineq2}
\end{align}
Bound the first term in \eqref{ineq2} using the inequality
\begin{align*}
\left| \frac{(1-\alpha_{12}) (P_{W}(x|y_1) -P_{W}(x|y_2))}{ \alpha_{12} P_{W}(x|y_1)+ (1-\alpha_{12}) P_{W}(x|y_2) }   
\right| 
\leq \frac{(1-\alpha_{12}) |P_{W}(x|y_1) -P_{W}(x|y_2)|}{ \alpha_{12} P_{W}(x|y_1)} 
\end{align*}
and do the same for the second term. We obtain the estimate

\begin{align*}
I(W)-I(\tilde{W})&\leq \frac{P_Y(y_1)}{\ln 2} \,\,\frac{1-\alpha_{12}}{\alpha_{12}} 
\sum_{x\in{\mathscr X}} |P_{W}(x|y_1) -P_{W}(x|y_2)| \\
&\quad+ \frac{P_Y(y_2)}{\ln 2} \,\,\frac{\alpha_{12}}{1-\alpha_{12}} 
\sum_{x\in{\mathscr X}} |P_{W}(x|y_1) -P_{W}(x|y_2)| \\
&= \frac{P_Y(y_1)+P_Y(y_2)}{\ln 2} ||P_{W}(.|y_1)-P_{W}(.|y_2)||_1.
\end{align*}
This completes the proof of \eqref{merge_ineq}.
\end{proof}

The bound \eqref{merge_ineq} brings in metric properties of the probability vectors. 
Leveraging them, we can use simple volume arguments to bound the rate loss due to approximation.

\begin{lemma}
Let the input and output alphabet sizes of $W$ be $q$ and $M$, respectively. Then,
there exists a pair of output symbols $(y_1,y_2)$ such that \textcolor{black}{for any $q\ge 2$}
   \begin{align}
 &P_Y(y_1)\leq \frac{2}{M}, \quad P_Y(y_2)\leq \frac{2}{M} , \label{eq:Y}\\
& ||P_{W}(.|y_1)-P_{W}(.|y_2)||_1\leq 
\textcolor{black}{
\frac{2}{\big(\frac M2\big)^{\frac 1{q-1}}-q}
}
\label{eq:pm}
\end{align}
which implies the estimate
\begin{equation}\label{eq:I}
0\leq I(W)-I(\tilde{W}) \leq
\textcolor{black}{
\frac {8}{\ln 2}\frac {1}{M\big( \big(\frac M2\big)^{\frac 1{q-1}}-q\big)}.}
\end{equation}
\label{second_lemma}
\end{lemma}

\begin{proof} Consider the subset of output symbols $A_M(\sY)=\{y: P_Y(y) \leq 2/M\}.$
Noticing that $|(A_M(\sY))^c| \leq M/2$, we conclude that
\begin{equation}
|A_M(\sY)| \geq \frac{M}{2}. \label{no_of_balls}
\end{equation}
Keeping in mind the bound \eqref{merge_ineq}, let us estimate 
the maximum value  of the quantity
\begin{equation}
\min_{y_1,y_2\in A_M(\sY) } \|P_{W}(\cdot|y_1)-P_{W}(\cdot|y_2)\|_1. \label{one_norm}
\end{equation}
For each $y\in {\mathscr Y}$, the vector $P_{W}(.|y)$ is an element
of the probability simplex 
   $$
   S_q= \left\{ (s_1,\dots,s_q)\in {\mathbb R}^q \bigg| s_i\geq0, \sum_{i=1}^q s_i= 1 \right\}.
   $$
Let $r>0$ be a number less than the quantity in \eqref{one_norm}. 
Clearly, for any $y_1,y_2\in A_M(\sY))$ the $q$-dimensional $\ell_1$-balls $B_{r/2,q}(P_{W}(\cdot|y_i))$
of radius $r/2$ centered at $P_{W}(\cdot|y_i), i=1,2$ are disjoint, and therefore, so are
their intersections with $S_q.$ 
It is easily seen\footnote{Indeed, let $
   \tilde S_q= \{ (s_1,\dots,s_q)\in {\mathbb R}^q \big| s_i\geq0, \sum_{i=1}^q s_i\le 1 \},
   $
then $\text{Vol}(\tilde S_q)=$ \textcolor{black}{ $\frac 1q\times\text{height}\times\text{base}$, where the height $h$ is the distance from $0$
to the base $S_q.$} We obtain $\text{Vol}(\tilde S_q)=\frac 1q h \text{Vol}(S_q),$ where
$h=1/\sqrt q.$ Finally, $\text{Vol}(\tilde S_q)$ is easily found by induction
to be $1/q!.$} that $\text{Vol}(S_q)=\sqrt q/(q-1)!$.

Our idea will be to estimate $r$ from above by a volume-type argument. This will give an upper bound on the smallest $\ell_1$ distance
in \eqref{one_norm}. Below we shorten the notation by writing $B_1(y_i):=B_{r/2,q}(P_{W}(\cdot|y_i)).$ Let
   $$
   T_q(r):=\Big\{(s_1,\dots,s_q)\in {\mathbb R}^q \Big| s_i\ge -\frac r2, i=1,\dots,q\Big\}.
   $$
By definition of $r$, for every $y_i$ we have $B_1(y_i)\subset T_q(r)$ (i.e., introducing $T_q(r)$ removes
the need to deal with the effects of the corner points of $S_q$). Let 
$H_q=\big\{(s_1,\dots,s_q)\in {\mathbb R}^q \big| \sum_{i=1}^q s_i=1\big\}.$ Since for different $i$ the balls $B_1(y_i)$ are 
pairwise disjoint, we have the obvious inequality
  \begin{equation}\label{eq:M2}
    \frac M2\le \frac{\text{Vol}(T_q(r)\cap H_q)}{\text{Vol}(B_1(y_i)\cap H_q)}
  \end{equation}
where the denominator does not depend on the choice of $i=1,2,\dots, q$. 

Let us compute the numerator in \eqref{eq:M2}. We can view $T_q(r)\cap H_q$ as a new simplex obtained by moving the origin
to the point $A=(-\frac r2, \dots,-\frac r2)$ and then drawing the straight lines parallel to the coordinate axes (see 
{Figure \ref{fig:simplex}}). The distance from $A$ to $H_q$ equals the distance from $A$ to the point $(\frac 1q,\dots,\frac 1q)$ which is   $\frac r2 \sqrt q+\frac 1{\sqrt q},$ and the distance from $0$ to $H_q$ equals the height of $S_q$ and is $\frac1{\sqrt q}.$ Therefore,
    $$
    \text{Vol}(T_q(r)\cap H_q)=\biggl(\frac{\frac r2\sqrt q+\frac1{\sqrt q}}{\frac 1{\sqrt{q}}}\biggr)^{q-1} \text{Vol}(S_q)=
    \Big(\frac r2 q+1\Big)^{q-1}\frac{\sqrt q}{(q-1)!}.
    $$

\begin{figure}[H]
\begin{center}\includegraphics[width=3in,angle=270]{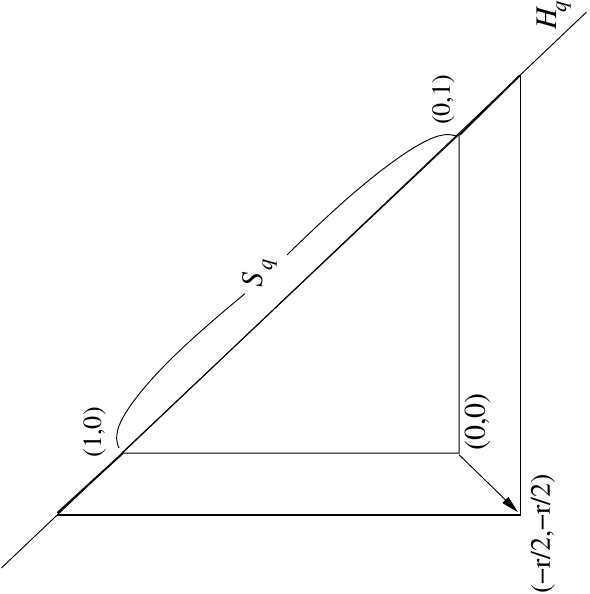}\end{center}
\caption{Expanding the simplex $S_q$ to compute a bound in \eqref{eq:M2}. This figure illustrates the case for $q=2.$}\label{fig:simplex}
\end{figure}

Now let us turn to the denominator in \eqref{eq:M2}. Note that $B_1(y_i)\supseteq B_2(y_i),$ where
$B_2(y_i)$ is an $\ell_2$ ball of radius $\frac r{2\sqrt q}$ centered at $P_{W}(\cdot|y_i).$ 
Further, $B_2(y_i)\cap H_q$ is an $\ell_2$ ball in $q-1$ dimensions, and its volume therefore is easily found:
  $$
  \text{Vol}(B_2(y_i)\cap H_q)=\frac{\pi^{\frac{q-1}2}}{{\Gamma}(\frac {q-1}2+1)}\Big(\frac {r}{2\sqrt{q}}\Big)^{q-1}.
  $$

We obtain
   $$
    \frac M2\le \frac{\text{Vol}(T_q(r)\cap H_q)}{\text{Vol}(B_1(y_i)\cap H_q)}\le 
    \frac{\text{Vol}(T_q(r)\cap H_q)}{\text{Vol}(B_2(y_i)\cap H_q)}=\frac{q^{\frac q2}}{(q-1)!}
    \frac{{{\Gamma}(\frac {q-1}2+1)}}{\pi^{\frac{q-1}2}} \Big(q+\frac 2r\Big)^{q-1}  $$
or, equivalently,
\begin{equation}\label{eq:2r}
   \frac 2r\ge   \Big(\frac{(q-1)!}{q^{\frac q2}}  \frac{\pi^{\frac{q-1}2}}{{\Gamma(\frac {q-1}2+1)}}
   \frac M2\Big)^{\frac 1{q-1}}-q\ge  \textcolor{black}{\Big(\frac M2\Big)^{\frac{1}{q-1}}-q}
\end{equation}
where the last inequality\footnote{\textcolor{black}{obtained by elementary calculations upon replacing $\Gamma(\cdot)$ with factorials.}} is valid for all $q\ge 2.$ From \eqref{eq:2r} we obtain
 \textcolor{black}{
   $$
   r\le \frac{2}{\big(\frac M2\big)^{\frac 1{q-1}}-q}.
   $$}
This proves  \eqref{eq:pm}, and \eqref{eq:I} follows immediately on applying Lemma \ref{first_lemma}. \end{proof}

This lemma leads to an important conclusion for the code construction: to degrade the subchannels
we should merge the symbols $y_1,y_2$ with small $P_Y(y_i)$ and such that the reverse channel
conditional PMFs $P_{W}(\cdot|y_i), i=1,2$ are $\ell_1$-close. Performing this step several
times in succession, we obtain the operation
called \texttt{degrade} in the description of Algorithm~\ref{euclid}.
The properties of this operation are stated in the following proposition.
\begin{proposition} Let $W$ be a DMC with input of size $q$.\\
(a) There exists
a function \texttt{degrade}$(W,\mu)$ such that its output channel $T$ satisfies
   \begin{equation}\label{eq:rateloss}
0\leq I(W)-I(T)\leq O\Big( \Big(\frac{1}{\mu}\Big)^{\frac 1{q-1}}\Big).
   \end{equation}
(b) For a given block length, let $W_N^{(i)}$ be the $i$-th subchannel after $n$ evolution steps of the polarization recursion.
Let $T_N^{(i)}$ denote the its approximation returned by Algorithm \ref{euclid}. 
Then
\begin{equation}\label{eq:vcrl}
0 \leq \frac{1}{N}\sum_{0\leq i \leq N} ( I( W_N^{(i)})- I( T_N^{(i)}) ) \leq n \, O\Big( \Big(\frac{1}{\mu}\Big)^{\frac 1{q-1}}\Big).
\end{equation}
\label{third_lemma}
\end{proposition}

\begin{proof}
Let $M$ be the cardinality of the output alphabet of $W$. Performing 
$M-\mu$ merging steps of the output symbols in succession, 
we obtain a channel with an output alphabet of size $\mu$. If the pairs of symbols to be merged
are chosen based on Lemma \ref{second_lemma}, then \eqref{eq:I} implies that
\begin{align*}
0&\leq I(W)-I(T)\\ &\leq C(q) \sum_{i=\mu+1}^M \left(\frac{1}{i}\right)^{\frac q{q-1}}\\
&\le C(q) \int_{\mu}^M (x-1)^{-({\frac q{q-1}})}dx\\
&=O\Big( \Big(\frac{1}{\mu}\Big)^{\frac 1{q-1}}\Big)
\end{align*}
where $C(q)$
is the constant \textcolor{black}{(implied by \eqref{eq:I})}, that depends on the input alphabet size $q$ but not on the number $n$
of recursion steps. 
\textcolor{black}{
Take $\mu$ large enough to satisfy
   $$
   \left(\frac{\mu}{2}\right)^{\frac{1}{q-1}}\!\!\!-q\geq \frac{1}{2}\left(\frac{\mu}{2}\right)^{\frac{1}{q-1}}
   $$
then from \eqref{eq:I} we see that the rate loss can be bounded above by
$$\frac{8}{\ln 2}
\frac{2}{\mu\left(\frac{\mu}{2}\right)^{\frac{1}{q-1}}}=\frac{16}{\ln 2}\cdot2^{\frac 1{q-1}}
\left(\frac{1}{\mu}\right)^{\frac q{q-1}}
.$$
Hence, we can take $C(q)=(\frac{16 }{\ln 2})2^{\frac 1{q-1}}.$}
This proves \eqref{eq:rateloss}, and \eqref{eq:vcrl} follows immediately.
\end{proof}

\begin{remark} This result provides a generalization to the nonbinary case of a result in  \cite{pedarsani}
which analyzed the merging (degrading) algorithm of \cite{tal_vardy}. For the case of binary-input channels,
Lemma 1 of \cite{pedarsani} gave an estimate $O(1/\mu)$ of the approximation error.
Substituting $q=2$ in \eqref{eq:rateloss}, we note that this result is a generalization of 
\cite{pedarsani} to channels with arbitrary finite-size input. \textcolor{black}{Arguing as in \cite{pedarsani}, we can claim that
$\mu=O(n^{2(q-1)})$ suffices to ensure that the error of approximation for most subchannels, apart from a vanishing
proportion of them, decays to zero.}
\end{remark}

\begin{remark} Upper bounds similar to \eqref{eq:rateloss} are derived in \cite[Lemma 6]{tal_sharov_vardy} and \cite[Lemma 8]{pereg_tal}.
The output symbol merging policy in \cite{tal_sharov_vardy} makes it possible to have
$I(W)-I(\tilde{W})=O( (1/\mu )^{1/q})$. On the other hand, the channel
upgrading technique introduced in \cite{pereg_tal} gives the same bound as \eqref{eq:rateloss}.
It is interesting to observe that merging a pair of output symbols at each step 
as we do here is as good as the algorithms based on binning of output symbols.
\end{remark}

\vspace{-0.0in}
\begin{remark} 
A very recent result of \cite{tal} states that any construction procedure of polar codes
construction based on degrading after each polarization step, that guarantees the rate loss bounded as $I(W)-I(T) \leq \epsilon,$
necessarily has the output alphabet of size $\mu=\Omega((1/\epsilon)^{\frac{q-1}{2}}).$
Proposition \ref{third_lemma}
implies that the alphabet size of the algorithm that we propose scales as the square of this bound,
meaning that the proposed procedure is not too far from being optimal,
namely for any channel, our degradation scheme satisfies $\mu \leq (1/\epsilon)^{q-1}$,
and there exists a channel for which $\mu \geq (1/\sqrt{\epsilon})^{q-1}$ holds true even for
the optimal degradation scheme. 
\end{remark}
\vspace{-0.1in}
\textcolor{black}{
\begin{remark}
The experimental results in Sect.~\ref{sim} indicate that the upper bound on the
average loss of symmetric capacity given by \eqref{eq:vcrl} is not
tight, and so it is likely that the example of the channel that accounts for the lower bound in \cite{tal} is 
an exception rather than the norm. 
We observe that the rate losses do not accumulate linearly with respect to $n$
in practice, and there is no need to choose $\mu$ to be polynomially dependent on $n$
in order to ensure a bounded rate loss. In our experiments we never take $\mu$ larger than 
$256$ and in many cases much smaller than that.
\label{rate_loss}
\end{remark}
}

\section{No-Loss Alphabet Reduction}
\label{cyclic}
Throughout this section we will use the transformation \textcolor{black}{\eqref{eq:+q}--\eqref{eq:-q}}, in which the ``$+$" is addition modulo $q.$ We discuss a way to further reduce the complexity of the code construction algorithm using 
the additive structure on $\sX.$  As shown in \eqref{merge_ineq}, the symmetric capacity loss is small if the posterior distributions 
induced by the merged symbols are $\ell_1$-close. Here we argue that if these vectors are related through  cyclic shifts, the output symbols can be merged \textcolor{black}{(using the merging operation defined later in equation \eqref{eq:ap}, which is different from the merging operation defined in Lemma \ref{lemma:defmerge})} at no cost to code performance.

Consider the construction of $q$-ary polar codes for channels with input alphabet $q\ge 2.$ Since $I(W)=\log q-H(X|Y),$ to construct polar codes it suffices to track
the values of $H(X|Y)$ for the transformed channels. Keeping in mind that $H(X|Y)=E(-\log P_{X|Y}(X|Y)),$ let us write 
the polarizing transformation in terms of the reverse channel 
$P_{X|Y}:$

\textcolor{black}{
 \begin{equation}\label{minustrans}
\left.\begin{aligned}
P_{Y^-}^{-}(y_i,y_j)&=P_Y(y_i)P_Y(y_j),\\
P_{X|Y^-}^{-}(x|y_i,y_j)&=\sum_{u_2\in\sX}P_{X|Y}(x\oplus u_2|y_i)P_{X|Y}(u_2|y_j)\\
P^-_X(x)&= \sum_{y_i,y_j\in\sY}P_{X|Y^-}^{-}(x|y_i,y_j)P_{Y^-}^{-}(y_i,y_j)\\
P_{Y^+}^{+ }(u,y_i,y_j)&=\left(\sum_{x\in\sX}P_{X|Y}(u\oplus x|y_i)P_{X|Y}(x|y_j)\right) P_Y(y_i) P_Y(y_j),\\
P_{X|Y^+}^{+ }(x|u,y_i,y_j)&=\frac{P_{X|Y}(u\oplus x|y_i)P_{X|Y}(x|y_j)}{\sum_{x_0\in\sX}P_{X|Y}(u\oplus x_0|y_i)P_{X|Y}(x_0|y_j)}\\
P^+_X(x)&= \sum_{u\in\sX,y_i,y_j\in\sY} P_{X|Y^+}^{+ }(x|u,y_i,y_j) P_{Y^+}^{+ }(u,y_i,y_j)
\end{aligned}\right\}
\end{equation}
}

If $P_X$ is uniform, both $P^+_X$ and $P^-_X$ are also uniform. \textcolor{black}{For this reason, 
 the posterior distributions $P_{X|Y^-}^{-}$ and $P_{X|Y^+}^{+ }$ defined in \eqref{minustrans}
are equal to the posterior distributions induced by the channels $W^-$ and $W^+$ defined in \textcolor{black}{\eqref{eq:+q}--\eqref{eq:-q}} respectively, under the uniform prior distributions.}
Throughout this section we will calculate the transformation of probability distributions using \eqref{minustrans} instead of 
\textcolor{black}{\eqref{eq:+q}--\eqref{eq:-q}} since we rely on the posterior distributions to merge symbols. 

\begin{definition} Given a distribution $P_{XY}$ on $\sX\times \sY,$ define an {\em equivalence relation} on $\sY$ as follows: $y_1\overset{P}{\sim} y_2$ if there exists $x_1\in\sX$ such that $P_{X|Y}(x\oplus x_1|y_1)=P_{X|Y}(x|y_2)$ for every $x\in\sX$. 
This defines a partition of $\sY$ into a set of equivalence classes $\cY=\{A_1,A_2,\dots,A_{|\cY|}\}.$
\end{definition}

We show that if $y_1\overset{P}{\sim} y_2,$ then we can \textcolor{black}{losslessly unify}  $y_1$ and $y_2$ into one alphabet symbol without 
changing $H(X|Y)$ for all $P_{XY}^s$, $s\in \{-,+\}^n$  and all $n\ge 1$. As a consequence, it is possible to assign one symbol to each equivalence class, i.e., the effective output alphabet of $W$
for the purposes of code construction is formed by the set $\cY$.

To formalize this intuition, we need the following definitions.

\vspace*{.05in}\begin{definition}\label{quid}
Consider a pair of distributions $P_{XY_1}, Q_{XY_2}.$  We say that two subsets of output alphabets $A\subseteq \sY_1, B\subseteq \sY_2$ are in correspondence, denoted $A\simeq B$, if 

(1) $P_{Y_1}(A)=P_{Y_2}(B)$;

(2) For every $y_1\in A$ and $y_2\in B$,  \textcolor{black}{there exists $x_1\in\sX,$ possibly depending on $y_1$ and $y_2$, such that for all 
$x\in\sX$, $P_{X|Y_1}(x\oplus x_1|y_1)=Q_{X|Y_2}(x|y_2)$.} 
\end{definition}
\vspace*{.05in}
Note that condition (2) in this definition implies that all the elements in $A$ are in the same equivalence class, 
and all the elements in $B$ are also in the same equivalence class.

\vspace*{.05in}
\begin{definition} \label{def:eq} We call the distributions  $P_{XY_1}, Q_{XY_2}$ {\em equivalent}, 
denoted $P_{XY_1}\equiv Q_{XY_2}$, if there is a bijection $\phi:\cY_1\to \cY_2$ such 
that $A\simeq \phi(A)$ for every equivalence class $A\in \cY_1.$ 
\end{definition}
\vspace*{.05in}

Note that two equivalent distributions have the same $H(X|Y).$ 

The following proposition underlies the proposed speedup of the polar code construction.
Its proof is computational in nature and is given in the Appendix.
\begin{proposition}\label{thm:simeq} Let $P_{XY_1}, Q_{XY_2}$ be two distributions. If $P_{XY_1}\equiv Q_{XY_2}$ then 
 for all $s\in\{-,+\}^n, n\ge 1$ we have $P^s_{XY_1^s}\equiv Q^s_{XY_2^s}$ (and therefore $H_{P^s}(X|Y_1)=H_{Q^s}(X|Y_2)$).
\end{proposition}

\begin{remark}
\textcolor{black}{Proposition \ref{thm:simeq} shows that the equivalence relation between two output symbols $y_1$
and $y_2$ is preserved after channel evolution steps, meaning that the symmetric capacity of all the subchannels
remains unchanged once $y_1$ and $y_2$ are unified.
Therefore,  the lossless unification introduced
in this section and the lossy merging described in Section \ref{main} can be used together and this does
not cause any complications in the sense that the analysis carried out in Section \ref{main} remains to be valid.}
\end{remark}

The next proposition provides a systematic way to \textcolor{black}{unify} output symbols of the synthesized channels obtained by the `$+$'
 transformation \textcolor{black}{in a lossless way.}

\begin{proposition}\label{novak}
Let distribution $P_{XY}$ on $\sX\times\sY,$ and let
$P^-_{XY^-}$ and $P^+_{XY^+}$ be defined as in \eqref{minustrans}.
For every $(v,y_1,y_2)\in\sX\times\sY^2$ we have 
\begin{equation}\label{plusone}
(v,y_1,y_2)\overset{P^+ }{\sim}(-v,y_2,y_1),
\end{equation}
where if $y_1=y_2$ then $v\ne 0.$
\end{proposition}
\begin{proof}
For every $y_1,y_2\in\sY$ and any $u_1,u\in\sX,$ we have
\begin{align*}
P^+_{X|Y^+} (u|(u_1,y_1,y_2))&=\frac{P_{X|Y}(u_1\oplus u|y_1)P_{X|Y}(u|y_2)}{\sum_{x_0\in\sX}P_{X|Y}(u_1\oplus x_0|y_1)P_{X|Y}(x_0|y_2)}\\
&=\frac{P(-u_1\oplus(u\oplus u_1)|y_2)P(u_1\oplus u|y_1)}{\sum_{x_0\in\sX}P_{X|Y}(-u_1\oplus (u_1\oplus x_0)|y_2)P_{X|Y}(u_1\oplus x_0|y_1)}\\
&=P^+_{X|Y^+} (u\oplus u_1|(-u_1,y_2,y_1)).
\end{align*}
This proves \eqref{plusone}.
\end{proof}

\underline{\sl No-loss cyclic \textcolor{black}{unification} algorithm}

\vspace*{.1in} Using the above considerations, we can reduce the time needed to construct a polar code. The informal
description of the algorithm is as follows.
Given a DMC $W:\sX\to\sY,$ we calculate a joint distribution $P_{XY}$ on $\sX\times \sY$ by assuming a uniform prior on $\sX.$ 
We then use \eqref{minustrans} to recursively calculate $P^s_{XY^s},$ and after each step of the recursion
we reduce the output alphabet size by assigning one symbol to the whole equivalence class. Namely,
for each equivalence class $A$ in the output alphabet $\sY^s,$ we set $P^s_{Y^s}(A)=\sum_{y\in A}P^s_{Y^s}(y)$ and $P^s_{X|Y^s}(x|A)=P^s_{X|Y^s}(x|y^*)$ for an arbitrarily chosen $y^*\in A.$ 
Note that $y^*$ can be chosen arbitrarily because the vectors $P^s_{X|Y^s}(\cdot|y), y\in A$ are cyclic shifts of each other.
By Prop. \ref{thm:simeq}, we have $I(W^s)=\log q-H_{P^s}(X|Y),$ i.e., 
the alphabet reduction entails no approximation of the capacity values.
 
Let us give an example, which shows that this simple proposal can result in a significant reduction of the size of the output alphabet.  Let $W$ be a {\em $q$-ary symmetric channel} ($q$SC) $W:\sX\to\sY,$ $|\sX|=|\sY|=q$
   \begin{equation}\label{eq:qsc}
W(y|x)=(1-\epsilon)\delta_{x,y}+\frac{\epsilon}{q-1}(1-\delta_{x,y}), 
  \end{equation}
\textcolor{black}{where $\delta_{x,y}$ is the Kronecker delta function},  
and let us take $q=4.$ Consider the channels $W^s,s\in\{+,-\}^n$ obtained by several applications of the recursion \textcolor{black}{\eqref{eq:+}--\eqref{eq:-}}. The actual output alphabet size of the channels $W^+,W^{++}$ and $W^{+++}$ is $4^3,4^7,$ and $4^{15},$ respectively. 
At the same time, the effective output alphabet size of $W^+,W^{++}$ and $W^{+++}$ obtained upon \textcolor{black}{unifying} the equivalence classes in $\sY$
is no more than $3,24,$ and $1200$ (the numbers come from experiment). In particular, the effective output alphabet size of $W^{+++}$ is less than a $10^{6}$-th fraction of its actual output alphabet size. Let $n\ge 3$ and $s\in\{+,-\}^n.$ If $s$ starts with $+++,$ then the effective output alphabet size of $W^s$ is less than a $(10^{6\times 2^{n-3}})$-th fraction of its actual alphabet size. 

\vspace*{.1in}
\underline{\sl Improved greedy mass merging algorithm}
\vspace*{.1in}

Now we are ready to describe the improved code construction scheme. 
Prop. \ref{thm:simeq} implies that if the vectors $P_{X|Y}(\cdot|y_i),i=1,2$ are cyclic shifts of each other,
\textcolor{black}{unifying them as} one symbol $\tilde y$ incurs no rate loss. Extending this intuition, we assume that performing
greedy mass merging using all the cyclic shifts of these vectors improves the accuracy of the approximation. 

Given a DMC $W:\sX\to\sY,$ we calculate a joint distribution $P_{XY}$ on $\sX\times \sY$ by assuming the uniform prior on $\sX$ and taking $W$ as the conditional probability. We then use \eqref{minustrans} to recursively calculate $P^s_{XY^s}$ and after each step of transformation:

(1) If the last step in $s$ is $+$: First use the \texttt{merge\textunderscore pair} function below to merge the symbols $(u_1,y_1,y_2)$ and $(-u_1,y_2,y_1)$ for all $u_1,y_1,y_2,$ then use the \texttt{degrade} function below on $P^s_{XY^s}.$

(2) If the last step in $s$ is $-$, use the \texttt{degrade} function below on $P^s_{XY^s}.$

The function \texttt{merge\textunderscore pair}$(Q,(y_1,y_2,u))$ is defined as follows:
Form the alphabet $\tilde \sY=\sY\backslash\{y_1,y_2\}\cup\{\tilde y\},$ 
putting $Q_{\tilde Y}(y)=Q_Y(y),Q_{X|\tilde Y}(x|y)=Q_{X|Y}(x|y)$ for all $x\in\sX$ and $y\in\tilde{\sY}\backslash \{\tilde y\}$ and
  \begin{equation}\label{eq:ap}
\begin{aligned}
Q_{\tilde Y}(\tilde y)&=Q_Y(y_1)+Q_Y(y_2),\\
  Q_{X|\tilde Y}(x|\tilde y)&=\frac{Q_Y(y_1)Q_{X|Y}(x|y_1)+Q_Y(y_2)Q_{X|Y}(x\oplus u|y_2)}{Q_{\tilde Y}(\tilde y)}.
\end{aligned}
  \end{equation}

\begin{remark}
Due to the concavity of the entropy function \cite[Thm. 2.7.3]{CoverThomas}, $H(X|Y)$ can only increase after calling the \texttt{merge\textunderscore pair} function.
\end{remark}

\begin{algorithm}
\caption{The \texttt{degrade} function  }\label{alg2}
\textbf{input:} distribution $P_{X,Y_0}$ over $\sX\times\sY_0,$ the target output alphabet size $\mu.$

\textbf{output:} distribution $Q_{X,Y}$ over $\sX\times\sY,$ where
$|\sY|\leq \mu.$
\begin{algorithmic}
\State $Q\gets P$
\State $\ell\gets |\sY|$
\While{$\ell>\mu$}
\State $(y_1,y_2,u)\gets \texttt{choose}(Q)$
\State $Q\gets \texttt{merge\textunderscore pair}(Q,(y_1,y_2,u))$
\State $\ell\gets \ell-1$
\EndWhile
\State \textbf{return} $Q$
\end{algorithmic}
\end{algorithm}

The function
\texttt{choose}$(Q)$ is defined as follows. 
Find the triple $y_1,y_2$ and $u\in{\mathscr X}$ such that
the change of conditional entropy $H_Q(X|Y)$ incurred by the merge $(y_1,y_2)\to \tilde y$ using \texttt{merge\textunderscore pair}$(Q,(y_1,y_2,u))$ 
    $$
   \Delta(H)\triangleq Q_{\tilde Y}(\tilde y)H(X|\tilde Y=\tilde y)-\sum_{i=1}^2 Q_Y(y_i)H(X|Y=y_i)
   $$
is the smallest among all the triples $(y_i,y_j,u)\in\sY^2\times\sX$. 

\begin{remark}{
The main difference between Algorithm~\ref{alg2} and the ordinary greedy mass merging algorithm discussed in Sect.~\ref{main} (e.g., 
Algorithm C in \cite{tal_vardy}) can be described as follows. In order to select a pair of symbols 
that induces the smallest increase of $H(X|Y),$ Algorithm~\ref{alg2} considers all the cyclic shifts of the posterior distributions of 
pairs of symbols, while the ``ordinary'' greedy mass merging algorithm examines only the distributions themselves. As argued above, this is the reason that Algorithm \ref{alg2} leads to a smaller rate loss than Algorithm \ref{euclid}.}

{Note that to perform the `+' transformation, we first use \eqref{plusone} to merge pairs of symbols with cyclically
shifted posterior vectors and then switch to greedy mass merging. In doing so, we incur a smaller rate loss because the number of
steps of approximation performed for Algorithm~\ref{alg2} is only half the number of steps performed in Algorithm \ref{euclid}.
Moreover, since \eqref{plusone} provides a systematic way of merging symbols with cyclically shifted distributions, (in other words, we do not need to search all the pairs in order to find them,) the running time of Algorithm \ref{alg2} is also reduced from that of Algorithm \ref{euclid}.
This intuition is confirmed in our experiments which show that the overall gap to capacity of the constructed codes is smaller than the one attained by using the basic greedy mass merging, while the time taken by the algorithm is reduced from greedy mass merging alone. More details about the experiments are given in Sect.~\ref{sim}. }
\end{remark}

\vspace*{.1in}
\underline{\sl The finite field transformation of Mori and Tanaka \cite{mori_tanaka}}
\vspace*{.1in}

We also note that Prop. \ref{thm:simeq} remains valid when the input alphabet of the channel is a finite field $\ff_q$ and Ar{\i}kan's transform $F=\text{\small{$\Big(\hspace*{-.05in}\begin{array}{c@{\hspace*{0.05in}}c}1&0\\[-.05in]1&1\end{array}\hspace*{-.05in}\Big)$}}$ is replaced by a transform based on the field structure, e.g., given by \eqref{eq:mt}.
This fact is stated in the following proposition whose proof is similar to Prop. \ref{thm:simeq} and will be omitted.
\begin{proposition}
Let $\sX=\ff_q$ and let $P_{XY_1}, Q_{XY_2}$ be two distributions. Suppose that the polarizing transform used is a finite-field 
type transform given by \eqref{eq:mt}. If $P_{XY_1}\equiv Q_{XY_2}$ then 
 for all $s\in\{-,+\}^n, n\ge 1$ we have $P^s_{XY_1^s}\equiv Q^s_{XY_2^s}$ (and therefore $H_{P^s}(X|Y_1)=H_{Q^s}(X|Y_2)$).
\end{proposition}
As a result of this statement, it is possible to define an accelerated construction procedure of codes over finite field alphabets 
similar to the algorithm discussed in this section.

\section{Experimental Results}
\label{sim}

There are several options of implementing the alphabet reduction procedures discussed above. The overall idea is to perform cyclic
merging (with no rate loss) and then greedy mass merging for every subchannel in every step $n\ge 1$ of the recursion.
\textcolor{black}{The experimental results show that the rate loss
$\Delta(I(W))$ does not grow linearly in $n$, and taking the output alphabet size $\mu$ a constant independent of $n$ is
sufficient to have a bounded rate loss.}

Greedy mass merging (the function \texttt{degrade} of Algorithm 1) calls for finding a pair of symbols
$y_1,y_2$ whose merging minimizes the rate loss $\tilde\Delta,$ which can be done in time $O(M^2\log M), M:=|\sY|.$ In practice this may be
too slow, so instead of optimizing we can merge the first pair of symbols for which the rate loss is below some chosen threshold $C.$
It is also possible to merge pairs of symbols based on the proximity of probabilities on the RHS of \eqref{merge_ineq}.

Note also that greedy mass merging can be applied to any binary polarizing operation including those described in Sect.~\ref{prelim}.
We performed a number of experiments using addition modulo $q$, the finite field polarization $G_\gamma,$ and a polarizing operation from
\cite{sasoglu4}. A selection of results appears in Fig.~\ref{fig1}.
In Examples 1-3 we construct polar codes for the $q$-ary symmetric channel \eqref{eq:qsc} and the 16 QAM channel, showing the distribution of capacities of the subchannels.
{In Examples 4-6 we apply different polarizing transforms to a channel $W$ with
inputs ${\mathscr X}=\{0,1\}^3$ and outputs ${\sY}=\{0,1\}^3\cup\{?\ast\ast,?\ast,???\},$
where $*$ can be 0 or 1. The transitions are given by
\begin{equation}\label{eq:OEC}
\begin{aligned}
W(x_1 x_2 x_3| x_1 x_2 x_3)=0.3, \;\;
W(? x_2 x_3| x_1 x_2 x_3)=0.2 \\
W(?? x_3| x_1 x_2 x_3)=0.3, \;\;
W(???|x_1x_2x_3)=0.2
\end{aligned}
\end{equation}
for all $x_1,x_2\in\{0,1\}$.
Following  \cite{park_barg}, we call $W$ an {\em ordered erasure channel}. One can observe that under the addition modulo-q transform 
\textcolor{black}{\eqref{eq:+q}--\eqref{eq:-q}} the channel polarizes to several extremal configurations, while under
the transforms given in \eqref{eq:mt}, \eqref{eq:pi} it converges to only two levels. This 
behavior, predicted by the general results cited in Section \ref{prelim}, supports the claim that
the basic algorithm of Sect.~\ref{main} does not depend on (is unaware of) the underlying polarizing transform.
More details about the experiments are provided in the captions to Fig.~\ref{fig1}.}
\textcolor{black}{For all our experiments both included here and left out, the rate loss
$\Delta (I(W))$ is less than $(1/\mu)^{1/(q-1)}$, let
alone less than $n(1/\mu)^{1/(q-1)}.$ This agrees with our claims in Remark \ref{rate_loss} above. }

It is interesting to observe that the $q$-ary symmetric channel for $q=16$ polarizes to two levels 
under Ar{\i}kan's transform. In principle there could be 5 different extremal configurations, and it is a priori unclear 
that no intermediate levels arise in the limit. An attempt to prove this fact was previously made in \cite{sasoglu3}, but
no complete proof is known to this date.

\setcounter{figure}{0}
\begin{figure}
\vspace*{.1in}\centering
\begin{subfigure}[t]{2.1in}
   \centering
   \includegraphics[width=2.0in]{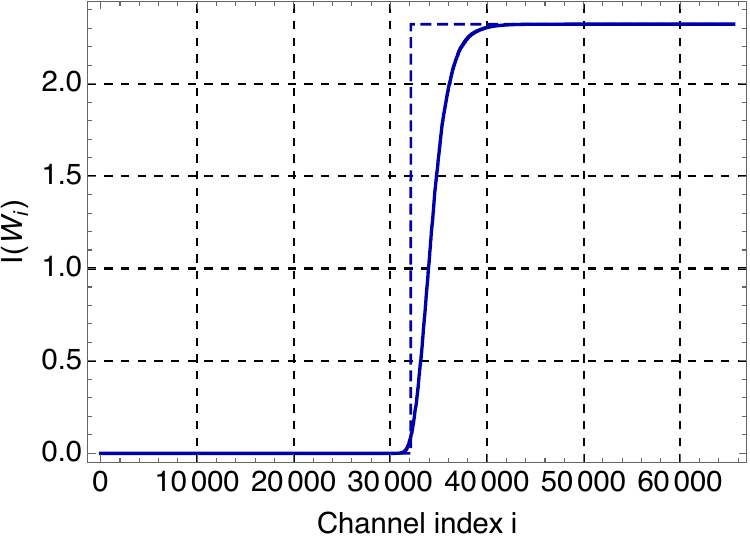}
   \caption{Example 1: $q$SC with $q=5$; $\epsilon=0.2$, $I(W)=1.2,$ 
     $n=16,$ \textcolor{black}{$\mu=200$,}
     $\Delta(I(W))=0.097$}
   \label{fig1:1}  
\end{subfigure}   
  \hspace*{.1in}
\begin{subfigure}[t]{2.1in}
    \centering{\includegraphics[width=1.9in]{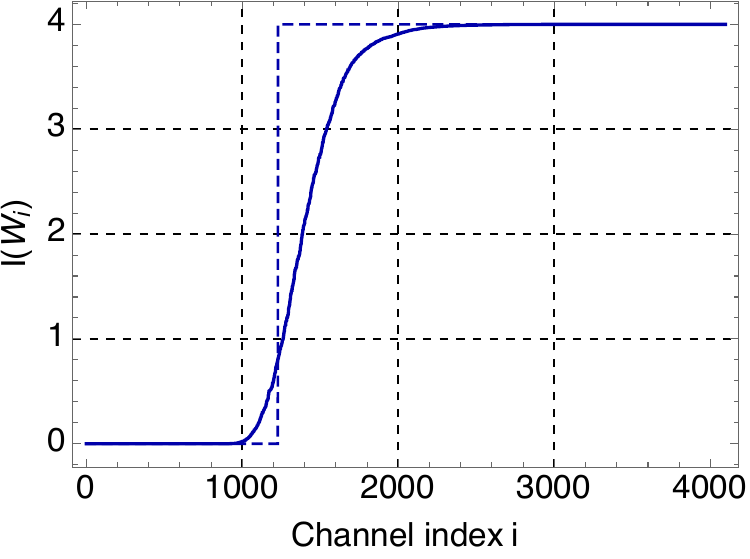}}   
 \caption{Example 2: 16 QAM, SNR=10dB, $I(W)=2.82$, $n=12,$ \textcolor{black}{$\mu=300$,} $\Delta(I(W))=0.2$}\label{fig1:2}
\end{subfigure}   
  \hspace*{.1in}
\begin{subfigure}[t]{2.1in}
   \centering{\includegraphics[width=1.9in]{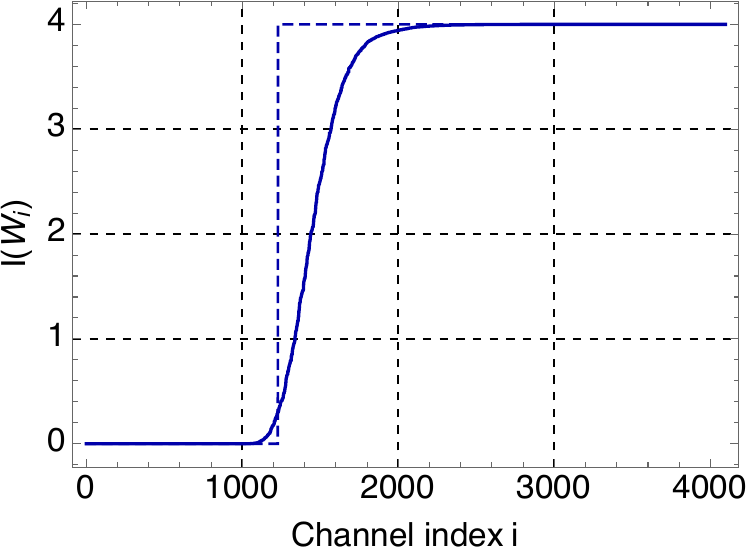}}
  \caption{Example 3: $q$SC with $q=16$; $\epsilon=0.15$, $I(W)=2.804,$ $n=12$,
  \textcolor{black}{$\mu=300$,}
  $\Delta(I(W))=0.23$}
    \label{fig1:3}  
\end{subfigure}  \\[0.05in] 

\begin{subfigure}[t]{2.1in}
  \centering\includegraphics[width=1.9in]{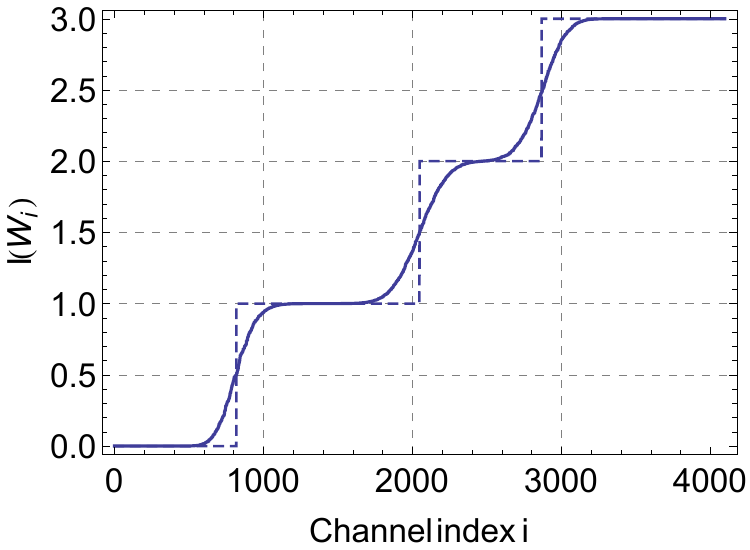}
  \caption{Example 4: OEC with $\epsilon_0=0.3$, $\epsilon_1=0.2$, $\epsilon_2=0.3$, $\epsilon_3=0.2$, $I(W)=1.6,$
   $n=12$, \textcolor{black}{$\mu=200$,}
  $\Delta(I(W))=0$ (in this case there is no approximation loss)
  }\label{fig1:8}
\end{subfigure} \hspace*{.1in}  
\begin{subfigure}[t]{2.1in}
  \centering\includegraphics[width=1.9in]{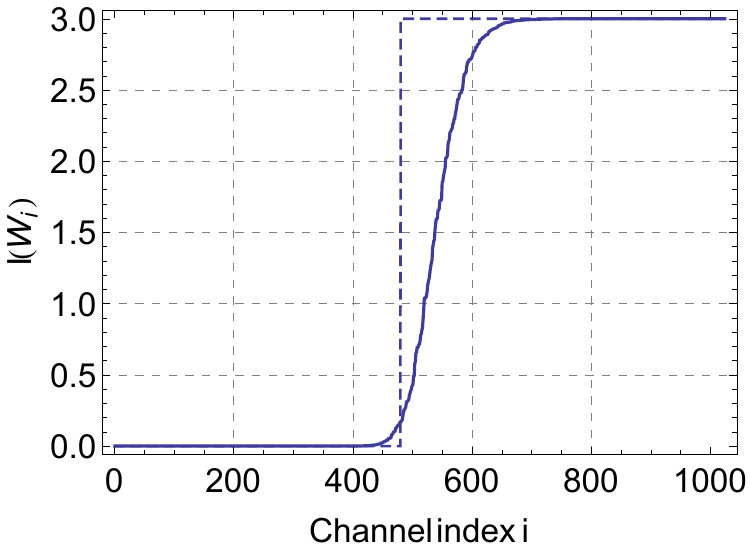}
  \caption{Example 5: The same channel as in Example 4, polarizing transform  \eqref{eq:pi}, $n=10,$
  \textcolor{black}{$\mu=200$,}
  $\Delta(I(W))=0.185$
  }\label{fig1:9}
\end{subfigure} \hspace*{.1in}  
\begin{subfigure}[t]{2.1in}
  \centering\includegraphics[width=1.9in]{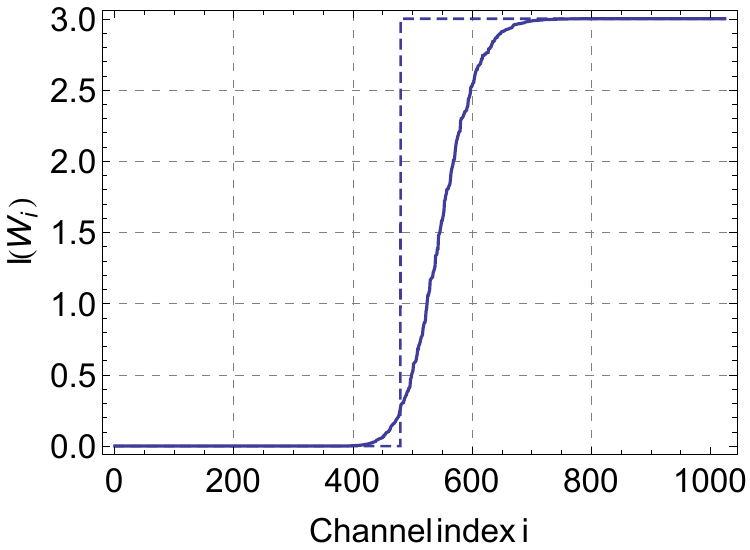}
  \caption{Example 6:
  The same channel as in Example 4, polarizing transform \eqref{eq:mt}, $n=10$,
  \textcolor{black}{$\mu=200$,}
  $\Delta(I(W))=0.216$ 
  }\label{fig1:10}
\end{subfigure} \hspace*{.1in}  
  
                  \caption{\small Construction of nonbinary polar codes. 
                  \textcolor{black}{The dashed step-shaped lines represent the ideal capacity distribution of the subchannels. The gaps between the dashed lines and
                  solid lines stem from unpolarized subchannels and the capacity loss due to channel degrading after each evolution step.}
                  In Fig. (a)-(c) we plot the 
                  capacity distribution of subchannels for channels with $q=5$ and $16$ (in these examples qSC is a $q$-ary symmetric channel defined in \eqref{eq:qsc}). \textcolor{black}{The results
                  suggest that the capacity lose increases when $q$ becomes larger.}
                  In Examples 4-6 we apply
                  different polarizing transforms, showing convergence to different number of extremal configurations for the same
                  channel (here OEC is the ordered erasure channel, see \eqref{eq:OEC})
                  }\label{fig1}
   \end{figure}
 
 Next we give some simulation results to support the conclusions drawn for Algorithm 2. 
We construct polar codes of several block lengths for $q$SC $W$ with $q=4$ and $\epsilon=0.15,$ setting the 
threshold $\mu=256.$ The capacity of the channel equals $I(W)=1.15242.$

{\small 
\begin{table}[h]
\begin{center}
\begin{tabular}{|c|c|c|c|c|c|c|}
\hline
$N$  & $t_1$  & $t_2$  & $\Delta I_1$  & $\Delta I_2$ & $\frac{t_1}{t_2}$ & $\frac{\Delta I_1}{\Delta I_2}$  \\[.03in]
\hline
128 & 404 & 177 & 0.041 & 0.026 & 2.3 & 1.6 \\
\hline
256 & 1038 & 490 & 0.048 & 0.033 & 2.1 & 1.5 \\
\hline
512 & 2256 & 1088 & 0.055 & 0.038 & 2.1 & 1.5 \\
\hline
1024 & 4378 & 2164 & 0.061 & 0.042 & 2.0 & 1.5 \\
\hline
\end{tabular}
\end{center}
\caption{\small The performance comparison of greedy mass
merging and Algorithm 2 for qSC with $q=4$, $\epsilon=0.15$, and $\mu=256.$}
\label{table1}
\end{table}}

In Table \ref{table1}, $N$ is the code length, $t_1$ is the running time of greedy mass merging and $t_2$ is the running
time of Algorithm 2 (our algorithm) in seconds. The quantities $\Delta I_1$ and $\Delta I_2$ represent the rate loss (the gap between $I(W)$ and the average capacity of the subchannels) in greedy mass merging and our algorithm, respectively. 

{\em Binary codes:}  Here our results imply the following speedup of Algorithm A in \cite{tal_vardy}.
Denote $LR(y)=W(y|1)/W(y|0).$ The cyclic merging means that we merge any two symbols $(y_1,y_2)\to \tilde y$ if
$LR(y_1)=LR(y_2)^{\pm 1},$ so we can only record the symbols $y\in \tilde Y$ with $LR(y)\ge 1.$
This implies that the threshold $\mu$ in \cite{tal_vardy} can be reduced to $\mu/2$. Overall the alphabet after the $+$ or $-$ step
is reduced by a factor of about $8$ while the code constructed is exactly the same as in \cite{tal_vardy}. 
In the following table we use the threshold values $\mu=32$ for \cite{tal_vardy} and $\mu=16$ for our algorithm. The codes are
constructed for the BSC with $\epsilon=0.11.$

{\small
\begin{table}
\begin{center}
\begin{tabular}{|c|c|c|c||c|c|c|c|}
\hline
$N$  & $t_A$  & $t_2$  & $\frac{t_A}{t_2}$ &$N$  & $t_A$  & $t_2$  & $\frac{t_A}{t_2}$  \\[.03in]
\hline
512 & 3.6 & 0.5 & 7.2 
&1024 & 7.3 & 1.1 & 6.6 \\
\hline
2048 & 14.7 & 2.3 & 6.4 
&4096 & 29.2 & 4.6 & 6.3 \\
\hline
\end{tabular}
\end{center}
\caption{\small The performance comparison of greedy mass
merging and Algorithm 2 for BSC(0.11).}
\label{table2}
\end{table}}

In Table \ref{table2}, $N$ is the code length, $t_A$ is the running time of Algorithm A in \cite{tal_vardy}, and $t_2$ is the running time of 
our algorithm in seconds. Our algorithm indeed is about 7 times faster, and the codes constructed in both cases are exactly the same.
 
\textcolor{black}{
\begin{remark}
The cyclic alphabet reduction for binary channels shares some similarity to the ideas introduced in \cite[Sect.~VI.C]{arikan2009}. At the same time, \cite{arikan2009} only observes the possibility of reducing the alphabet size without giving a practical alphabet reduction algorithm, while our algorithm can be readily implemented.
\end{remark}
}

\section{Conclusion}
\label{conclusion}

We considered the problem of constructing polar codes
for nonbinary alphabets. Constructing polar codes has been a difficult open question since the introduction
of the binary polar codes in \cite{arikan2009}. Ideally, one would like to
obtain an explicit description of the polar codes for a given block length, but this seems
to be beyond reach at this point. As an alternative, one could attempt to construct the code by approximating
each step of the recursion process. For binary codes, this has been done in \cite{tal_vardy},\cite{pedarsani}, but
extending this line of work to the nonbinary case was an open problem despite several attempts in the literature.
We take this question one step closer to the solution by designing an algorithm that approximates
the construction for moderately-sized input alphabets such as $q=16.$
The algorithm we implement works for both binary and non-binary channels with complexity $O(N\mu^2 \log \mu)$,
where $N$ is the blocklength and $\mu$ is the parameter that limits the output alphabet size.
Furthermore, the error estimate the we derive generalizes the estimate of \cite{pedarsani} 
to the case of nonbinary input alphabets (but relies on a different proof method).
It is also interesting to note that the error is rather close to a {\em lower bound} for this type
of construction algorithms, derived recently in \cite{tal}.
Apart from presenting a theoretical advance, this algorithm provides a useful tool in the
analysis of properties of various polarizing transforms applied to nonbinary codes over
alphabets of different structure. The proposed construction algorithm also brings nonbinary 
codes closer to practical applications, which is another promising direction to be explored
in the future.

\section*{Appendix: Proof of Prop.~\ref{thm:simeq} }
We will show that if $P_{XY_1}\equiv Q_{XY_2},$ then $P^-_{XY_1^-}\equiv Q^-_{XY_2^-}$ and $P^+_{XY_1^+}\equiv Q^+_{XY_2^+},$
which will imply the full claim by induction on $n.$\\

(a) (The `$-$' case) The distributions $P^-_{XY_1^-}$ and $Q^-_{XY_2^-}$ are defined on the sets 
$\sX\times \sY_1^2$ and $\sX\times\sY_2^2,$ respectively. In order to prove 
that $P^-_{XY_1^-}\equiv Q^-_{XY_2^-},$ we need to show that for every $A_1,B_1\in\cY_1,$ we have $A_1\times B_1\simeq \phi(A_1)\times \phi(B_1).$ Indeed, 
\begin{align*}
\sum_{(y_1,y_2)\in A_1\times B_1}P^-_{Y_1^-}((y_1,y_2))&=\sum_{y_1\in A_1}\sum_{y_2\in B_1}P^-_{Y_1^-}((y_1,y_2))\\
&=\sum_{y_1\in A_1}\sum_{y_2\in B_1}P_{Y_1}(y_1)P_{Y_1}(y_2)\\
&=\Big(\sum_{y_1\in A_1}P_{Y_1}(y_1)\Big)\Big(\sum_{y_2\in B_1}P_{Y_1}(y_2)\Big).
\end{align*}
Similarly,
$$
\sum_{(y_1,y_2)\in \phi(A_1)\times \phi(B_1)}Q^-_{Y_2^-}((y_1,y_2))=\Big(\sum_{y_1\in \phi(A_1)}Q_{Y_2}(y_1)\Big)\Big(\sum_{y_2\in \phi(B_1)}Q_{Y_2}(y_2)\Big).
$$
Since $A_1\simeq \phi(A_1)$ and $B_1\simeq \phi(B_1),$ we have $P_{Y_1}(A)=Q_{Y_2}(\phi(A)_1)$ and
$P_{Y_1}(B)=Q_{Y_2}(\phi(B)_1).$
Therefore,
   $$
\sum_{(y_1,y_2)\in A_1\times B_1}P^-_{Y_1^-}((y_1,y_2))=\sum_{(y_1,y_2)\in \phi(A_1)\times \phi(B_1)}Q^-_{Y_2^-}((y_1,y_2)).
   $$
Thus $A_1\times B_1$ and $\phi(A_1)\times \phi(B_1)$ satisfy condition (1) in Def. \ref{quid}. 

To prove condition (2), choose $y_1\in A_1,y_2\in B_1,$ and let $y_3\in \phi(A_1)$ and $y_4\in\phi(B_1).$ By Def. \ref{quid}, there exist $x_1$ and $x_2$ such that
$P_{X|Y_1}(x\oplus x_1|y_1)=Q_{X|Y_2}(x|y_3)$ and $P_{X|Y_1}(x\oplus x_2|y_2)=Q_{X|Y_2}(x|y_4)$ for all $x\in\sX.$ Thus
\begin{align*}
Q^-_{X|Y_2^-}(x|(y_3,y_4))&=\sum_{u_2\in\sX}Q_{X|Y_2}(x\oplus u_2|y_3)Q_{X|Y_2}(u_2|y_4)\\
&=\sum_{u_2\in\sX}P_{X|Y_1}(x\oplus u_2\oplus x_1|y_1)P_{X|Y_1}(u_2\oplus x_2|y_2)\\
&=\sum_{u_2\in\sX}P_{X|Y_1}((x\oplus z)\oplus u_2|y_1)P_{X|Y_1}(u_2|y_2)\\
&=P^-_{X|Y_1^-}(x\oplus z|(y_1,y_2)),
\end{align*}
where $z=x_1\oplus(- x_2)$.
Therefore, $A_1\times B_1\simeq \phi(A_1)\times \phi(B_1),$ and $P_{XY_1}\equiv Q_{XY_2}.$\\

(b). ({\em The `$+$' case}) The distribution $P^+_{XY_1^+}$ and $Q^+_{XY_2^+}$ are over 
$\sX\times(\sX\times\sY_1^2)$ and $\sX\times(\sX\times\sY_2^2)$ respectively. 
Similarly to case (a) above, we will show that for every $A_1,B_1\in\cY_1$ there exist permutations 
$\pi_{y_1,y_2}$ and $\pi_{y_3,y_4}$ on $\sX$ such that for every $u\in \sX$
 $$
   \{(\pi_{y_1,y_2}(u),y_1,y_2):y_1\in A_1,y_2\in B_1\}\simeq \{(\pi_{y_3,y_4}(u),y_3,y_4):y_3\in \phi(A_1),y_4\in \phi(B_1)\}
 $$ 
To show this, fix $A_1,B_1\in\cY_1$ and choose 
some $z_1\in A_1,z_2\in B_1,y_1\in A_1,y_2\in B_1,$ $y_3\in \phi(A_1)$ and $y_4\in\phi(B_1).$
By Def. \ref{quid},  for every $x\in \sX$ there exist $x_1,x_2,x_3$ and $x_4$ such that
    \begin{gather*}
    P_{X|Y_1}(x\oplus x_1|z_1)=P_{X|Y_1}(x|y_1), \quad 
    P_{X|Y_1}(x\oplus x_2|z_2)=P_{X|Y_1}(x|y_2)\\
    P_{X|Y_1}(x\oplus x_3|z_1)=Q_{X|Y_2}(x|y_3),\quad
    P_{X|Y_1}(x\oplus x_4|z_2)=Q_{X|Y_2}(x|y_4).
    \end{gather*}
For $x\in \sX$ define permutations $\pi_{y_1,y_2},\pi_{y_3,y_4}$ as $\pi_{y_1,y_2}(x)=-x_1\oplus x\oplus x_2$ 
and $\pi_{y_3,y_4}(x)=-x_3\oplus x\oplus x_4.$  We compute
      \begin{align*}
    P^+_{X|Y_1^+}(x|(\pi_{y_1,y_2}(u),y_1,y_2))&=P^+_{X|Y_1^+}(x|(-x_1\oplus x_2 \oplus u,y_1,y_2))\\
     &=\frac{P_{X|Y_1}(-x_1 \oplus x_2 \oplus x\oplus u|y_1)P_{X|Y_1}(x|y_2)}
       {\sum_{x_0\in\sX}P_{X|Y_1}(-x_1 \oplus x_2 \oplus x_0\oplus u|y_1)P_{X|Y_1}(x_0|y_2)}\\
     &=\frac{P_{X|Y_1}(x\oplus x_2\oplus u|z_1)P_{X|Y_1}(x\oplus x_2|z_2)}
          {\sum_{x_0\in\sX}P_{X|Y_1}(x_0\oplus x_2\oplus u|y_1)P_{X|Y_1}(x_0\oplus x_2|y_2)}\\
     &=P^+_{X|Y_1^+}(x\oplus x_2|(u,z_1,z_2)).
       \end{align*}
Similarly,
    $$
Q^+_{X|Y_2^+}(x|(\pi_{y_3,y_4}(u),y_3,y_4))=P^+_{X|Y_1^+}(x\oplus x_4|(u,z_1,z_2)).
    $$
The last two equations imply that
     $$
P^+_{X|Y_1^+}(-x_2\oplus x_4\oplus x|(\pi_{y_1,y_2}(u),y_1,y_2))=Q^+_{X|Y_2^+}(x|(\pi_{y_3,y_4}(u),y_3,y_4)),
     $$
which verifies condition (2) in Def.~\ref{quid}.
Let us check that condition (1) is satisfied as well. We have
       \begin{align*}
     P^+_{Y_1^+}(\{(\pi_{y_1,y_2}(u),y_1,y_2): &\,y_1\in A_1,y_2\in B_1\})
                      =\sum_{y_1\in A_1,y_2\in B_1}P^+_{Y_1^+}((\pi_{y_1,y_2}(u),y_1,y_2))\\
    =&\sum_{y_1\in A_1,y_2\in B_1}P_{Y_1}(y_1)P_{Y_1}(y_2)\sum_{x\in\sX}
                P_{X|Y_1}(- x_1\oplus x_2\oplus x\oplus u|y_1)P_{X|Y_1}(x|y_2)\\
    =&\sum_{y_1\in A_1,y_2\in B_1}P_{Y_1}(y_1)P_{Y_1}(y_2)\sum_{x\in\sX}P_{X|Y_1}(u\oplus x_2\oplus x|z_1)
                P_{X|Y_1}(x\oplus x_2|z_2)\\
    =&\Big(\sum_{x\in\sX}P_{X|Y_1}(u\oplus x|z_1)P_{X|Y_1}(x|z_2)\Big)
               \Big(\sum_{y_1\in A_1}P_{Y_1}(y_1)\Big)\Big(\sum_{y_2\in B_1}P_{Y_1}(y_2)\Big)
      \end{align*}
and
\begin{align*}
Q^+_{Y_2^+}(\{(\pi_{y_3,y_4}(u),y_3,y_4):&\;y_3\in \phi(A_1),y_4\in \phi(B_1)\})\\
=&\Big(\sum_{x\in\sX}P_{X|Y_1}(u\oplus x|z_1)P_{X|Y_1}(x|z_2)\Big)\Big(\sum_{y_3\in \phi(A_1)}Q_{Y_2}(y_3)\Big)\Big(\sum_{y_4\in \phi(B_1)}Q_{Y_2}(y_4)\Big).
\end{align*}
     By assumption $P_{Y_1}(A_1)=Q_{Y_2}(\phi(A_1))$ and $P_{Y_1}(B_1)=Q_{Y_2}(\phi(B_1)),$ so this proves that
$$
P^+_{Y_1^+}(\{(\pi_{y_1,y_2}(u),y_1,y_2):y_1\in A_1,y_2\in B_1\})=Q^+_{Y_2^+}(\{(\pi_{y_3,y_4}(u),y_3,y_4):y_3\in \phi(A_1),y_4\in \phi(B_1)\}).
$$
Thus for every $u\in \sX$
    $$
    \{(\pi_{y_1,y_2}(u),y_1,y_2):y_1\in A_1,y_2\in B_1\}\simeq \{(\pi_{y_3,y_4}(u),y_3,y_4):y_3\in \phi(A_1),y_4\in \phi(B_1)\}
    $$ 
The proof is complete.

\end{document}